\documentclass[12pt]{article}

\usepackage{graphicx}
\usepackage{amsmath}
\usepackage{amssymb}
\usepackage{amsthm}
\usepackage{xcolor}
\usepackage{hyperref}
\usepackage{algpseudocode}
\usepackage{algorithm}
\usepackage{authblk}

\newcommand{\be}{\begin{equation}}
\newcommand{\ee}{\end{equation}}
\newcommand{\ba}{\begin{array}}
\newcommand{\ea}{\end{array}}
\newcommand{\bea}{\begin{eqnarray}}
\newcommand{\eea}{\end{eqnarray}}

\newcommand{\ra}{\rangle}
\newcommand{\la}{\langle}

\newcommand{\ket}[1]{{\left\vert{#1}\right\rangle}}

\newcommand{\calL}{{\cal L }}

\newcommand{\calC}{{\cal C }}
\newcommand{\calS}{{\cal S }}

\newcommand{\calP}{{\cal P }}

\newcommand{\ZZ}{\mathbb{Z}}

\newtheorem{dfn}{Definition}

\newtheorem{lemma}{Lemma}
\newtheorem{fact}{Fact}

\newtheorem{theorem}{Theorem}

\newcommand{\eq}[1]{Eq.~(\ref{eq:#1})}
\renewcommand{\sec}[1]{\hyperref[sec:#1]{Section~\ref*{sec:#1}}}
\newcommand{\ssec}[1]{\hyperref[ssec:#1]{Subsection~\ref*{ssec:#1}}}
\newcommand{\fig}[1]{\hyperref[fig:#1]{Fig.~\ref*{fig:#1}}}
\newcommand{\tab}[1]{\hyperref[tab:#1]{Table~\ref*{tab:#1}}}
\newcommand{\lem}[1]{\hyperref[lem:#1]{Lemma~\ref*{lem:#1}}}
\newcommand{\prop}[1]{\hyperref[prop:#1]{Proposition~\ref*{prop:#1}}}
\newcommand{\thm}[1]{\hyperref[thm:#1]{Theorem~\ref*{thm:#1}}}

\newcommand{\Gate}[1]{\textsc{#1}}
\newcommand{\hgate}{\Gate{H}}

\newcommand{\tgate}{\Gate{T}}
\newcommand{\pgate}{\Gate{P}}

\newcommand{\notgate}{\Gate{NOT}}
\newcommand{\cnotgate}{\Gate{CNOT}}
\newcommand{\toffoligate}{\Gate{Toffoli}}
\newcommand{\fredkingate}{\Gate{Fredkin}}
\newcommand{\swapgate}{\Gate{SWAP}}

\begin{document}
\title{Efficient ancilla-free reversible and quantum circuits \\ 
for the Hidden Weighted Bit function}

\author{Sergey Bravyi, Theodore J. Yoder, and Dmitri Maslov}
\affil{IBM Quantum, IBM T. J. Watson Research Center \\ Yorktown Heights, NY 10598, USA}

\maketitle

\abstract{
The Hidden Weighted Bit function plays an important role in the study of classical models of computation.  A common belief is that this function is exponentially hard for the implementation by reversible ancilla-free circuits, even though introducing a small number of ancillae allows a very efficient implementation. In this paper we refute the exponential hardness conjecture by developing a polynomial-size reversible ancilla-free circuit computing the Hidden Weighted Bit function.  Our circuit has size $O(n^{6.42})$, where $n$ is the number of input bits. We also show that the Hidden Weighted Bit function can be computed by a quantum ancilla-free circuit of size $O(n^2)$. The technical tools employed come from a combination of Theoretical Computer Science (Barrington’s theorem) and Physics (simulation of fermionic Hamiltonians) techniques.
}

\section{Introduction}
The origins of the  Hidden Weighted Bit function go back to the study of models of classical computation.  This function, denoted $\mathsf{HWB}$,  takes as input an $n$-bit string $x$ and outputs the $k$-th bit of $x$, where $k$ is the Hamming weight of $x$; if the input weight is $0$, the output is $0$.  It is best known for combining the ease of algorithmic description and implementation by classical Boolean circuits with the hardness of representation by Ordered Binary Decision Diagrams (OBDDs) \cite{bryant1986graph}---a popluar tool in VLSI \cite{meinel2012algorithms}.  The difference between logarithmic-depth implementations of $\mathsf{HWB}$ by circuits (recall that $\mathsf{HWB} \,{\in}\, NC^1$ but $\mathsf{HWB} \,{\not\in}\, AC^0$) and an exponential lower bound for the size of the OBDD \cite{bryant1991complexity} is startling two exponents.  Relaxing the constraints on the type of Binary Decision Diagram considered or restricting the computations by circuits enables a multitude of implementations with polynomial cost \cite{bollig1999complexity}.

The Hidden Weighted Bit function was first introduced in the context of reversible and quantum computations about 15 years ago by I. L. Markov and K. N. Patel (unpublished), and the earliest explicit mention dates to the year 2005 \cite{maslov2005toffoli}.  The original specification is irreversible, and required a slight modification to comply with the restrictions of reversible and quantum computations.  Specifically, the Hidden Weighted Bit function was redefined to become the cyclic shift to the right by the input weight. We denote this reversible specification as $\mathbf{hwb}$.  
Formally, $\mathbf{hwb}(x)$ is defined as the cyclic shift of its input $x$ to the right by $W$ positions, where $W\,{=}\,x_1{+}x_2{+}\ldots {+}x_n$ is the Hamming weight of $x$.  The following shows the truth table of $3$-input $\mathbf{hwb}$:
\begin{center}
\begin{tabular}{r|c|c|c|c|c|c|c|c}
$x$ & 000 & 100 & 010 & 110 & 001 & 101 & 011 & 111 \\
\hline
$\mathbf{hwb}(x)$ & 000 & 010 & 001 & 101 & 100 & 011 & 110 & 111\\
\end{tabular}
\end{center}
Since its introduction, $\mathbf{hwb}$ was used by numerous authors focusing on the synthesis and optimization of reversible and quantum circuits as a test case. 

Despite a stream of improvements in the respective circuit sizes by various research groups \cite{prasad2006data, maslov2007techniques, donald2008reversible, saeedi2010reversible}, the best known ancilla-free reversible circuits exhibit exponential scaling in the number of gates. The synthesis algorithms benefiting from the inclusion of additional gates, such as multiple-control multiple-target Toffoli, Fredkin, and Peres gates \cite{maslov2005toffoli, donald2008reversible, maslov2005synthesis} also failed to find an efficient implementation without ancillae.  In 2013, this culminated with the $\mathbf{hwb}$ receiving the designation of a ``hard'' benchmark function \cite{saeedi2013synthesis}.  A recent asymptotically optimal synthesis algorithm over the library with $\notgate$, $\cnotgate$, and $\toffoligate$ gates \cite{zakablukov2016application}, introduced in the year 2015, was also unable to find an efficient ancilla-free implementation. 
An ancilla-free quantum circuit can be obtained by employing an asymptotically optimal quantum circuit synthesis algorithm such as \cite{shende2006synthesis}, but the quantum gate count appears to remain exponential and larger than what is possible to obtain through the application of the asymptotically optimal reversible logic synthesis algorithm \cite{zakablukov2016application}.

The introduction of even a small number of ancillae changes the picture dramatically. Just $O(\log(n))$ ancillary (qu)bits suffice to develop a reversible circuit with $O(n \log^2(n))$ gates \cite{maslov2005reversible}. Barrington's theorem \cite{barrington1989bounded} allows one to obtain a polynomial-size reversible circuit using three ancillae.  This polynomial-size three-ancilla reversible circuit can be obtained by computing the individual bits of the input weight through Barrington's theorem, and using such bits logarithmically many times to control-$\swapgate$ the respective input (qu)bits into their desired positions.  Finally, the existence of a polynomial-size quantum circuit using a single ancilla follows from \cite{ablayev2005computational}.

State of the art, in both the classical reversible and quantum settings, thus points to an exponential difference in the gate count between circuits with no ancillae and circuits with a constant number of ancillae.  In this paper, we demonstrate efficient implementations of the $\mathbf{hwb}$ function by ancilla-free reversible and quantum circuits, thereby reducing these exponential differences to polynomial.  Specifically, our reversible ancilla-free circuit requires $O(n^{6.42})$ gates and our quantum ancilla-free circuit requires $O(n^2)$ gates. 
These results refute the exponential hardness belief and remove $\mathbf{hwb}$ from the class of hard benchmarks.

We next sketch main ideas behind our ancilla-free circuits. We begin with the reversible circuit.  Our construction works as follows.  First, we show that the $n$-bit $\mathbf{hwb}$ function can be decomposed into a product of $O(n\log(n))$ gates denoted $C5(f(x);B)$, where $f(x)$ is a 
symmetric Boolean function and $B\,{\subset}\,x$ is a subset with $5$ input bits. The gate $C5(f;B)$ cyclically shifts the $5$-bit register $B$ if $f(x){=}1$, and does nothing when $f(x){=}0$.  To implement $C5(f;B)$, we first break it down into a product of $6$ gates of the form $C5|_{M_i}(f(x{\setminus}B);B)$, where $i\,{\in}\,\{1,2,3,4,5,6\}$, each $M_i$ is a fixed set of Boolean $5$-tuples, and $f$ are symmetric Boolean functions. The gate $C5|_{M_i}$ restricts the operation of the corresponding gate $C5$ onto the set $M_i$ and simultaneously separates the set $B$ of bits being cycle-shifted from the set $x{\setminus}B$ controlling these shifts.  This allows to employ Barrington's theorem~\cite{barrington1989bounded} to implement the gates $C5|_{M_i}(f(x{\setminus}B);B)$ in the ancilla-free fashion by expressing them as polynomial-size branching programs with the input $x{\setminus}B$ and computing into $B$.  Each instruction in such program realizes a permutation of $5$-bit strings controlled by a single bit and it can thus be mapped into a reversible circuit over $6\,{=}\,5{+}1$ wires. 

Next we introduce our quantum ancilla-free circuit.  Let $U_{\mathbf{hwb}}$ be the $n$-qubit unitary operator implementing the $\mathbf{hwb}$ function. By definition, $U_{\mathbf{hwb}}|x\rangle = \mathsf{C}^{x_1+x_2+\ldots+x_n}|x\ra$, where $\mathsf{C}$ is the cyclic shift of $n$ qubits.  Suppose we can find an $n$-qubit Hamiltonian $H$ such that $\mathsf{C}\,{=}\,e^{iH}$ and $H$ commutes with the Hamming weight operator $W\,{=}\,\sum_{j=1}^n |1\ra\la 1|_j$.  Then $U_{\mathbf{hwb}}=e^{iHW}$.  Thus it suffices to construct a quantum circuit simulating the time evolution under the Hamiltonian $HW$.  Since the cyclic shift $\mathsf{C}$ is analogous to the translation operator for a particle moving on a circle, the Hamiltonian $H$ generating the cyclic shift $\mathsf{C}$ is analogous to the particle's momentum operator.  This observation suggests that $H$ can be diagonalized by a  suitable Fourier transform.  We formalize this intuition using the language of fermions and the fermionic Fourier transform, which is routinely used in Physics and quantum simulation algorithms \cite{babbush2017low, kivlichan2018quantum}.  The desired Hamiltonian $H$ such that $\mathsf{C}\,{=}\,e^{iH}$ is shown to have the form $H\,{=}\,V^\dag H' V$, where $V$ is a (modified) fermionic Fourier transform and $H'$ is a simple diagonal Hamiltonian.  We also show that $V$ commutes with the Hamming weight operator $W$, so that $U_{\mathbf{hwb}}\,{=}\,e^{iHW}\,{=}\,V^\dag e^{iH'W} V$.  We demonstrate that each layer in this decomposition of $U_{\mathbf{hwb}}$ can be implemented by a quantum circuit of size $O(n^2)$.

The rest of the paper is organized as follows.  \sec{revanc} introduces a simple modification of the known $O(n\log^2(n))$-gate $O(\log(n))$-ancilla reversible circuit that requires $O(n\log(n))$ gates and $O(\log(n))$ ancillary bits. \sec{revnoa} describes an $O(n^{6.42})$-gate ancilla-free reversible circuit. \sec{quantumnsquared} reports an ancilla-free $O(n^2)$-gate quantum circuit. 
These sections are independent of each other and can be read in any order.
Appendices~A and~B prove technical lemmas stated in \sec{quantumnsquared}.

\section{Reversible circuit of size $O(n\log n)$ using ancillas}
\label{sec:revanc}

We start with the description of a modification of the previously reported classical/reversible circuit that implements $\mathbf{hwb}$ with $O(n\log(n))$ gates and $O(\log(n))$ ancillae \cite{maslov2005reversible}.  Compared to \cite{maslov2005reversible}, our circuit features favorable asymptotics. However, it uses twice the computational/ancillary space.

Similarly to \cite{maslov2005reversible}, we break down the computation into three stages:
\begin{enumerate}
    \item Compute the input weight $W=x_1{+}x_2{+}\ldots{+}x_n$. 
    \item Apply controlled-$\swapgate$ gates to SWAP inputs into their correct position as specified by the $\mathbf{hwb}$. 
    \item Restore the value of ancillary register to $\ket{0}$ by appending the inverse of the stage 1.
\end{enumerate}
Note that the stage 3. is omitted in \cite{maslov2005reversible}, allowing a direct comparison to our circuit illustrated in \fig{hwb7}. The difference between our construction and \cite{maslov2005reversible} is how we compute the input weight.  Specifically, we use the same ``plus-one'' approach to calculate the weight into the ancillary register, however, we implement the integer increment function differently. Given input $x_i$, $1 {\leq} i {\leq} n$, the resister $w_1,w_2,\ldots,w_{\lfloor\log(i)\rfloor{+}1}$, where the input weight is being computed into, and temporary storage $t_1,t_2,\ldots,t_{\lfloor\log(i)\rfloor{-}1}$, ``increment by one'' works as follows. If $i\,{=}\,1$, apply $\cnotgate(x1;w1)$. For $i\,{>}\,1$:
\begin{enumerate}
    \item if $i{>}3$: apply Toffoli gate to $\ket{x_i,w_1,t_1}$; for $j$ from $2$ to $\lfloor\log(i)\rfloor{-}1$ apply the Toffoli gate $\toffoligate(t_{j-1},w_j;t_j)$;
    \item if $i\,{=}\,2$ or $i\,{=}\,3$ apply $\toffoligate(x_j,w_1;w_2)\cnotgate(x_j;w_1)$; 
    \newline else apply $\toffoligate(t_{\lfloor\log(i)\rfloor-1},w_{\lfloor\log(i)\rfloor};w_{\lfloor\log(i)\rfloor+1})\cnotgate(t_{\lfloor\log(i)\rfloor-1};w_{\lfloor\log(i)\rfloor})$.
    \item if $i{>}3$: for $j$ from $\lfloor\log(i)\rfloor{-}1$ down to $2$ apply the half adder, computed by the circuit $\toffoligate(t_{j-1},w_j;t_j)\cnotgate(t_{j-1};w_j)$. Apply $\toffoligate(x_i,w_1;t_1)\cnotgate(x_i;w_1)$.
\end{enumerate}
In our implementation, the $t$ register is used to store necessary digit shifts.  Advertised asymptotics follow by inspection of the above construction.  We furthermore illustrated our circuit in \fig{hwb7} for $n{=}7$.  
 
\begin{figure}[t]
\centerline{\includegraphics[height=4cm]{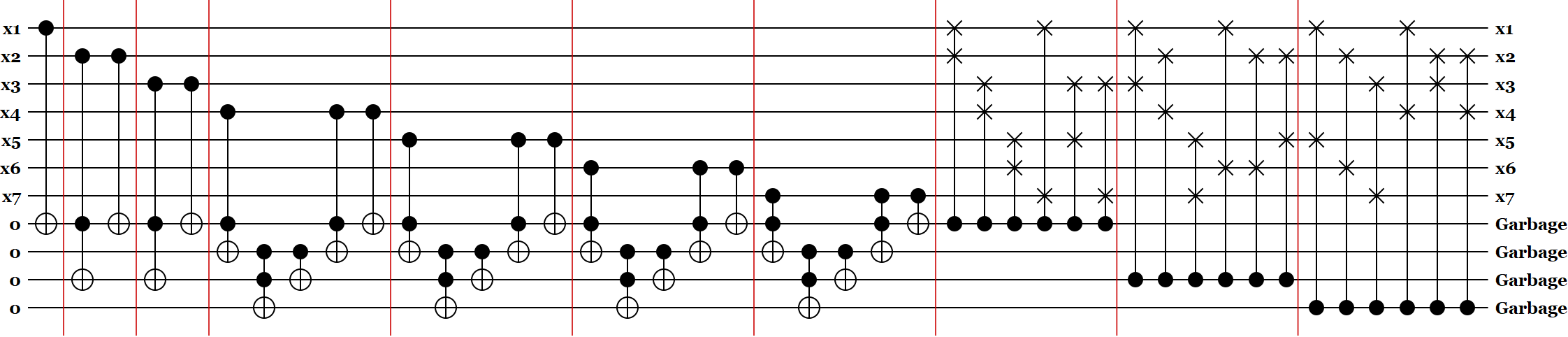}}
\caption{10-stage reversible circuit applying the $7$-bit $\mathbf{hwb}$ to
$\ket{x_1x_2x_3x_4x_5x_6x_7w_1t_1w_2w_3}$.  
Each of first $7$ $\cnotgate/\toffoligate$ gate stages increments $\ket{w_1w_2w_3}$ by one depending on the value of input variable, next $3$ $\fredkingate$ gate stages perform controlled-$\swapgate$.  Vertical red lines separate these $10$ stages. Not shown is Garbage uncomputation that can be performed by appending the inversion of the weight calculation circuit ($\cnotgate/\toffoligate$ gate part).}\label{fig:hwb7}
\end{figure}

\section{Ancilla-free reversible circuit of size $O(n^{6.42})$}\label{sec:revnoa}

In this section we show how to construct an ancilla-free classical reversible circuit of size $poly(n)$ implementing $\mathbf{hwb}$. We focus on $n\,{\geq}\,5$, noting that optimal circuits with $n$ up to $4$ are already known. 

Let $n$ be the total number of bits, and $x=(x_1,x_2,...,x_n)\in\{0,1\}^n$ be the input. In some discussions where it is convenient, we label these bits by the integers $\{0,1,\ldots,n{-}1\} \,{=}\, \ZZ_n$.  Suppose $B\,{\subseteq}\, \ZZ_n$ is a subset of $5$ bits and $f: \{0,1\}^n\to \{0,1\}$ is a symmetric Boolean function
(that is, $f(x)$ depends only on the Hamming weight of $x$).
Define a reversible gate 
\[
C5(f;B): \, \{0,1\}^n \to \{0,1\}^n,
\]
where the output is obtained from the input $x$ by applying the cyclic shift to the register $B$ if $f(x){=}1$.  Otherwise, when $f(x){=}0$, the gate does nothing.  Note that, because the symmetric function $f$ does not depend on the order of the bits, $C5(f;B)$ is a permutation of the set $\{0,1\}^n$.  Moreover, $C5(f;B)$ is an even permutation, since it is a product of length-$5$ cycles and each length-$5$ cycle is an even permutation.

Define $C(f;(i_0,i_1,...,i_{t-1}))$ to be a reversible gate that applies the cyclic shift of some $t$ bits defined by the cycle $(i_0,i_1,...,i_{t-1})$  (where $i_0,i_1,...,i_{t-1} \in \ZZ_n$ are all distinct) if the symmetric function $f$ evaluates to one and does nothing otherwise. We call $i_0,i_1,...,i_{t-1}$ the targets.  We call a collection of $C$-type gates a layer when the sets of their targets do not overlap. 

We next construct $\mathbf{hwb}$ by first expressing it as a circuit with the $C$-type gates, then breaking down the $C$-type gates into elementary reversible gates and $C5$-type gates, and finally expressing the $C5$-type gates in terms of the elementary reversible gates.

\begin{lemma}\label{lem:hwbwithc}
The $n$-bit $\mathbf{hwb}$ function can be implemented by an ancilla-free circuit with $\lfloor\log(n)\rfloor+1$ layers of $C$-type gates. 
\end{lemma}
\begin{proof}
We will create a circuit with $k$ layers numbered $0,1,...,\lfloor\log(n)\rfloor$.  At each layer, the $C$ gates take the form $C(f_k;*)$.  Select the symmetric functions $f_k$ as follows: let $f_k(x)\,{=}\,1$ iff the $k$th power of $2$ in the binary expansion of the weight $W=x_1{+}x_2{+}...{+}x_n$ equals one.  Note that $f_k$ are symmetric functions since the calculation of weight does not depend on the order the bits are added in. The function $\mathbf{hwb}$ can now be expressed as 
\begin{equation}\label{eq:hwblogC}
\mathbf{hwb} = C^{2^0}((0,1,...,n{-}1);f_0) C^{2^1}((0,1,...,n{-}1);f_1) \cdots C^{2^{\lfloor\log(n)\rfloor}}((0,1,...,n{-}1);f_{\lfloor\log(n)\rfloor}).
\end{equation}

For any $k=0,1,\dots,\lfloor\log(n)\rfloor$, let $g:=\text{GCD}(n,2^k)$ and $C_{i}:=C(f_k;(i,\,i{+}2^k \bmod n,\,...,\linebreak i{+}(\frac{n}{g}{-}1)2^k \bmod n))$. Then by elementary modular arithmetic,
\[
C^{2^k}((0,1,...,n{-}1);f_k)=C_{0}C_{1}\dots C_{g-1},
\]
and the targets of any two distinct $C_{i}$ in this product do not overlap. This shows that each of the $\lfloor\log(n)\rfloor{+}1$ factors in Eq.~\eqref{eq:hwblogC} can be written as a layer of $C$-type gates.
\end{proof}

\begin{figure}[t]
	\centerline{\includegraphics[height=4cm]{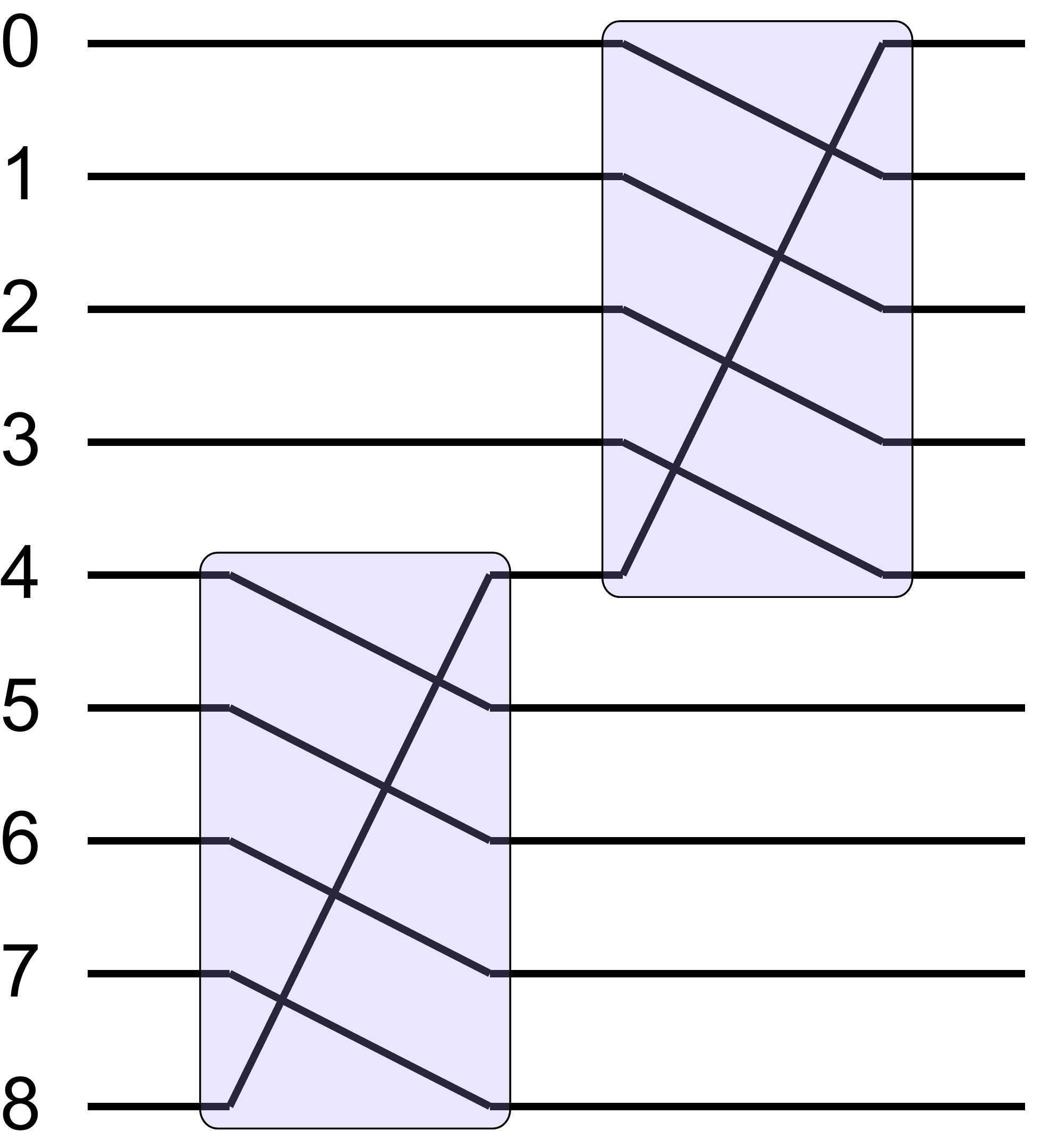}}
	\caption{Implementation of the $9$-bit cyclic shift $C(f;(0,1,2,3,4,5,6,7,8))$ using the gates $C5(f;(4,5,6,7,8))$ and $C5(f;(0,1,2,3,4))$.}
	\label{fig:cycles}
\end{figure}

We next implement each of $\lfloor\log(n)\rfloor{+}1$ layers of cyclic shift gates in \lem{hwbwithc} as circuits with $O(n)$ $C5$-type gates by expressing the cycles $(i_0,i_1,...,i_{t-1})$ as products of length-$5$ cycles.  Note that a length-$5$ cycle is always an even permutation and $(i_0,i_1,...,i_{t-1})$ is an odd permutation when $t$ is even.  It is not possible to implement an odd permutation as a product of even permutations. However, with one exception, the $C$-type gates $C_i$ come in pairs (recall that their number, $g$, is a power of two) and thus they can usually be paired up to form an even permutation that can then be decomposed into a product of length-$5$ cycles. The one exception is the leftmost gate in 
Eq.~\eqref{eq:hwblogC},
$C(f_0;(0,1,...,n{-}1))$, when $n$ is even.  We handle this case first.

\begin{figure}[t]
\centerline{\includegraphics[height=3.8cm]{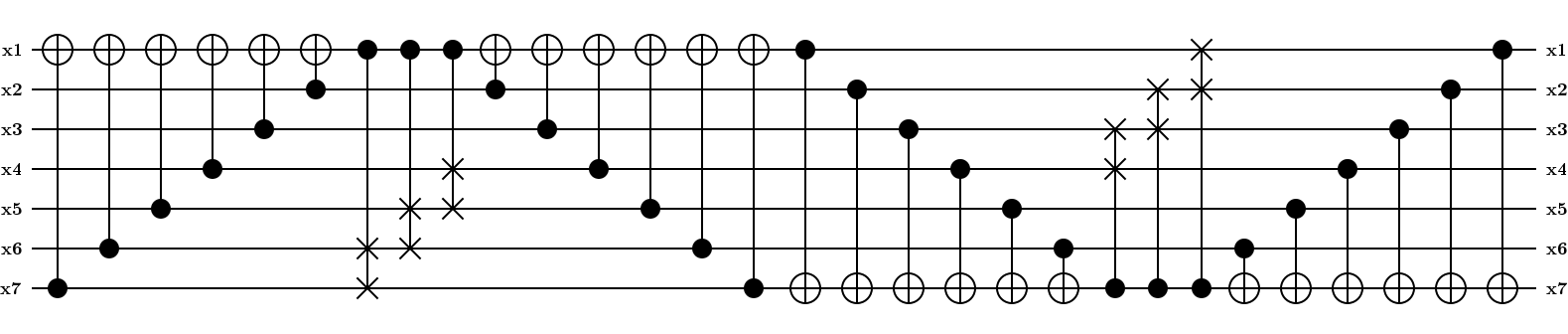}}
\caption{Implementation of $C(f_0;(x_1,x_2,x_3,x_4,x_5,x_6,x_7))$, where $f_0(x) = x_1 \oplus x_2 \oplus x_3 \oplus x_4 \oplus x_5 \oplus x_6 \oplus x_7$.}\label{fig:7bitC0}
\end{figure}

\begin{lemma}\label{lem:C0}
$C(f_0;(0,1,...,n{-}1))$ can be implemented by a reversible circuit with $O(n)$ elementary gates.
\end{lemma}
\begin{proof}
The Boolean function $f_0(x) = x_1 \oplus x_2 \oplus ... \oplus x_n$ can be implemented on the top bit to control all bit SWAPs on the bottom bits, and it can be implemented on the bottom bit to control all bit SWAPs on the top bits. The number of controlled-$\swapgate$ gates required is $n{-}1$, and the total number of the $\cnotgate$ gates required to compute/uncompute the control register is $4(n{-}1)$.  We illustrated this construction in \fig{7bitC0} for $n{=}7$.
\end{proof}

\begin{lemma}\label{lem:lemC}
For $n\,{\geq}\,5$: 
\begin{enumerate}
    \item for $t\,{\leq}\,4$, pairs of two $C(f;(i_0,i_1,...,i_{t-1}))$ gates can be implemented by an ancilla-free circuit using constantly many gates $C5(f;B)$; 
    \item for odd $t\,{>}\,4$ the $C(f;(i_0,i_1,...,i_{t-1}))$ gate can be implemented by an ancilla-free circuit using $O(t)$ gates $C5(f;B)$.
    \item for even $t\,{>}\,4$ pairs of $C(f;(i_0,i_1,...,i_{t-1}))$ gates can be implemented by an ancilla-free circuit using $O(t)$ gates $C5(f;B)$;
\end{enumerate}
\end{lemma}

\begin{proof}\,

\noindent{\bf 1.} There are three cases to consider: $t\,{=}\,2$, $t\,{=}\,3$, and $t\,{=}\,4$. 
\begin{itemize}
    \item[$t\,{=}\,2$.] $C(f;(x_1,x_2))$ and $C(f;(y_1,y_2))$ can be implemented simultaneously by the circuit \linebreak $C5(f;(y_1,x_1,y_2,a,x_2))C5(f;(a,y_1,x_1,y_2,x_2))$.  This is equivalent to saying that the following permutation equality holds: $(x_1,x_2)(y_1,y_2) = (y_1,x_1,y_2,a,x_2)(a,y_1,x_1,y_2,x_2)$. Note that the bit `$a$' can be found since $n\,{\geq}\,5$. We will show only the permutation equalities in the rest of the proof, since it is trivial to translate those to circuits. 
    \item[$t\,{=}\,3$.] To implement a pair of gates $C(f;(x_1,x_2,x_3))$ and $C(f;(y_1,y_2,y_3))$ rely on the cycle product equality $(x_1,x_2,x_3)(y_1,y_2,y_3) = (x_1,y_1,x_2,y_2,y_3)(x_3,x_1,y_1,x_2,y_2)$.
    \item[$t\,{=}\,4$.] Cycles $(x_1,x_2,x_3,x_4)$ and $(y_1,y_2,y_3,y_4)$ can be obtained by the equality
    \begin{eqnarray*}
    (x_1,x_2,x_3,x_4)(y_1,y_2,y_3,y_4) \\
    = (x_1,x_2)(x_1,x_3,x_4) \cdot (y_1,y_2)(y_1,y_3,y_4) \\
    = (x_1,x_2)(y_1,y_2) \cdot (x_1,x_3,x_4)(y_1,y_3,y_4),
    \end{eqnarray*}
    where first and second part require two $C5$ gates each, as described in the cases $t\,{=}\,2$ and $t\,{=}\,3$, for a total of four $C5$ gates.
\end{itemize}

\noindent{\bf 2.} The goal is to develop a circuit with $C5$ gates implementing the gate $C(f;(0,1,...,t{-}1))$, where $t$ is odd. There are two cases to consider, $t\,{=}\,4p{+}1$ and $t\,{=}\,4p{+}3$.

\noindent\underline{Case 1: $t\,{=}\,4p{+}1,\;p\geq 1$.} We want to implement the integer permutation given by the cyclic shift $(0,1,...,4p)$ by the cyclic shifts of length $5$.  
This can be done as follows, 
\[
(0,1,...,4p) = (4p-4,4p-3,4p-2,4p-1,4p) (4p-8,4p-7,4p-6,4p-5,4p-4) \cdots (0,1,2,3,4).
\]
This decomposition uses $p$ length-5 cycles, resulting in the ability to implement $C(f;(0,1,...,t{-}1))$ gate using $p{=}\frac{t-1}{4}$ $C5(f;B)$ gates. This construction is illustrated in \fig{cycles} for $n\,{=}\,9$.

\noindent\underline{Case 2: $t\,{=}\,4p{+}3,\;p\geq 1$.} Use the formula
\begin{eqnarray*}
(0,1,...,4p{+}2) 
= (4p,4p{+}1,4p{+}2)\cdot(0,1,...,4p) 
\\ = (4p{+}2,4p,2,1,0)(4p{+}1,4p{+}2,0,1,2) \cdot (0,1,...,4p).
\end{eqnarray*}
Since we already implemented $(0,1,...,4p)$ with $p$ $C5$ gates in Case 1 above, this implementation requires $\frac{t+5}{4}$ $C5$ gates. 

\noindent{\bf 3.} The goal is to implement a pair of  $C(f;(x_1,x_2,...,x_t))$ and $C(f;(y_1,y_2,...,y_t))$ where $t\,{>}\,4$ is even. Write 
\begin{eqnarray*}
(x_1,x_2,...,x_t) \cdot (y_1,y_2,...,y_t) \\
= (x_1,x_2)(x_1,x_3,x_4,...,x_t) \cdot (y_1,y_2)(y_1,y_3,y_4,...,y_t) \\
= (x_1,x_2)(y_1,y_2) \cdot (x_1,x_3,x_4,...,x_t) \cdot (y_1,y_3,y_4,...,y_t)
\end{eqnarray*}
Here, $(x_1,x_2)(y_1,y_2)$ requires two $C5$ gates per item {\bf 1.} case $t\,{=}\,2$, and each of $(x_1,x_3,x_4,...,x_t)$ and $(y_1,y_3,y_4,...,y_t)$ requires $O(t)$ gates per item {\bf 2.}
\end{proof}

Observe how the above proof implies that the number of $C5$ gates required to implement each of $C^{2^k}((0,1,...,n{-}1);f_k)$ stages in \eq{hwblogC} for $k=1,2,...,\lfloor\log(n)\rfloor$ is between $\frac{n}{4}+Const$ and $\frac{n}{2}+Const$. Thus, per \lem{C0}, the total number of elementary and $C5$ gates required to implement $\mathbf{hwb}$ over $n$ qubits is between $\frac{n \log(n)}{4} + O(n)$ and  $\frac{n \log(n)}{2} + O(n)$.

We next show how to implement $C5(f_k;B)$ as a branching program, using Barrington's theorem~\cite{barrington1989bounded}, by closely following the original proof.  In preparation for using Barringon's theorem, we first remove the dependence of the functions $f_k$ in $C5(f_k;B)$ on the variables inside the set $B$, to allow the desired cyclic shift to be controlled by the values of $n{-}5$ variables outside the set $B$ itself.  To accomplish this, note that $C5(f_k;B)$ acts trivially on the strings $00000$ and $11111$; those can be ignored.  This leaves $30$ non-fixed by the operation 5-bit strings that can be partitioned into six disjoint subsets $M_1,M_2,M_3,M_4,M_5$, and $M_6$, with $5$ strings each.  Every subset $M_i$ contains $5$ cyclic shifts of some fixed $5$-bit string, and is defined as follows:
\begin{align} \label{eq:Ms}
M_1 :=& \{10000,01000,00100,00010,00001\}, \\
M_2 :=& \{01111,10111,11011,11101,11110\}, \nonumber \\
M_3 :=& \{11000,01100,00110,00011,10001\}, \nonumber  \\
M_4 :=& \{10100,01010,00101,10010,01001\}, \nonumber \\ 
M_5 :=& \{00111,10011,11001,11100,01110\}, \nonumber\\
M_6 :=& \{01011,10101,11010,01101,10110\}. \nonumber
\end{align}
We implement $C5(f_k;B)$ by performing the cyclic shifts of a single subset $M_i$ per time. 


First, let us introduce some more notations. Given a bit string $x\,{\in}\, \{0,1\}^n$, write $x=(y,b)$, where $b\,{\in}\, \{0,1\}^5$ is the restriction of $x$ onto the register $B$ and $y\,{\in}\, \{0,1\}^{n-5}$ is the rest of $x$.  Let $w_i\,{\in}\, \{1,2,3,4\}$ be the Hamming weight of bit strings in $M_i$ (note that all strings in the same subset $M_i$ have the same weight).  Define a Boolean function $f_{k,i} : \{0,1\}^{n-5} \to \{0,1\}$ such that $f_{k,i}(y){=}1$ iff $2^k$ appears in the binary expansion of $|y|{+}w_i$. Then 
\[
f_k(x)=f_k(y,b)=f_{k,i}(y) \quad \mbox{for any } b\in M_i.
\] 
Define a gate
\[
C5|_{M_i}(f_k;B) : \, \{0,1\}^n \to \{0,1\}^n
\]
that maps an input $x\,{=}\,(y,b)$ to an output $x'\,{=}\,(y,b')$ according to the following rules:
\begin{itemize}
    \item if $f_{k,i}(y)\,{=}\,0$ then $b'\,{=}\,b$;
    \item if $f_{k,i}(y)\,{=}\,1$ and $b\,{\notin}\,M_i$ then $b'\,{=}\,b$;
    \item if $f_{k,i}(y)\,{=}\,1$ and $b\,{\in}\, M_i$ then $b'\,{\in}\, M_i$ is obtained from $b$ by cyclically shifting the elements of $M_i$.
\end{itemize}
By definition, the cyclic shift of bits in the register $B$ can be realized by cyclically shifting elements of each subset $M_i$ for $i=1,2,3,4,5,6$.  Thus
\begin{eqnarray}\label{eq:GBk}
C5(f_k(x);B) = \prod_{i=1}^6 C5|_{M_i}(f_k(y);B).
\end{eqnarray}
Here the order in the product does not matter because the gates $C5|_{M_i}(f_k;B)$ pairwise commute.  Note that the dependence of function $f_k$ on the variables inside the set $B$ has now been removed, and we can proceed to implementing $C5|_{M_i}(f_k;B)$ as a branching program, and finally mapping the instructions used by the branching program into reversible gates.

Recall some relevant notation used in Barrington's paper \cite{barrington1989bounded}. Let
$S_5$ be the group of permutations of $5$ numbers, $\{1,2,3,4,5\}$.  
Given a $5$-tuple of distinct integers $a_1,a_2,a_3,a_4,$ and $a_5$, we write $(a_1,a_2,a_3,a_4,a_5)$ to denote the $5$-cycle. Let $e$ be the identity permutation.
A branching program of length $L$ with $m$ Boolean input variables $y_1,y_2,...,y_m$  is a list of instructions $\langle y_i, \sigma_i, \tau_i \rangle$ with $i=1,2,...,L$ and $\sigma_i, \tau_i \in S_5$, such that $\sigma_i$ is applied if $y_i{=}1$, and $\tau_i$ is executed when $y_i{=}0$.
Given a permutation  $\sigma \in S_5$, the branching program is said to $\sigma$-compute a Boolean function $f(y)$ if executing the list of all instructions in the program results in $e$ (the identity permutation) for all inputs $y$ such that $f(y){=}0$ and permutation $\sigma$ for all inputs $y$ such that $f(y){=}1$.

Barrigton's theorem asserts that any function in the class $\mathrm{NC}^1$ can be $(1,2,3,4,5)$-computed by a branching program of polynomial size \cite{barrington1989bounded}.  We next specialize the proof of the theorem to explicitly develop a short branching program that $(1,2,3,4,5)$-computes the Boolean function $f_{k,i}(y)$. Recall that $f_{k,i}(y){=}1$ iff $2^k$ appears in the binary expansion of $y_1{+}y_2{+}\ldots{+}y_{n-5}{+}w_i$ with $w_i\in \{1,2,3,4\}$ being the weight of bit strings in $M_i$. It suffices to develop a branching program computing the Boolean function $f_k(y)$ with $y\in \{0,1\}^m$ and $m\,{=}\,n{-}5$ by appending at most two constant binary variables $1$ encoding $w_i$ to the bit string $y$.

While the original proof \cite{barrington1989bounded} explored the mapping of logarithmic-depth classical circuits over $\{\mathsf{AND}, \mathsf{OR}\}$ library, we focus on the classical circuits over 3-input 1-output  $\mathsf{MAJ}(a,b,c):=ab \oplus bc \oplus ac$ and $\mathsf{XOR}(a,b,c):= a \oplus b \oplus c$ gates.  Recall that the library $\{\mathsf{MAJ}, \mathsf{XOR}\}$ is universal for classical computations if constant inputs are allowed.

\begin{lemma}
\label{lem:S5}
Suppose $y$ is an $m$-bit string and $f_k(y)$ is the $k$-th bit in the binary representation of $W=y_1{+}y_2{+}...{+}y_m$. The function $f_k(y)$ can be $(1,2,3,4,5)$-computed by a branching program of size $O(m^{5.42})$.
\end{lemma}
\begin{proof}
First, we describe a logarithmic-depth classical circuit that computes functions $f_k(y)$ for the range of applicable values $k$, and second, report expressions for $\mathsf{MAJ}$ and $\mathsf{XOR}$ in the form of a branching program that can be used in the recursion \cite[Proof of Theorem 1]{barrington1989bounded}.  The length of the branching program computing $f_k(y)$ is upper bounded by taking the maximal length of the program implementing $\mathsf{MAJ}$ or $\mathsf{XOR}$ to the power of the circuit depth.

First, construct a classical circuit with $\mathsf{MAJ}$ and $\mathsf{XOR}$ gates that implements $f_k(y)$.  To do so, we develop a circuit that computes all bits of the $W(y)$, and for the purpose of implementing a given single Boolean component, discard all gates that compute the bits we are not interested in.   
Such operation does not increase the depth of the circuit, and may, in fact, decrease it slightly.

To find $W(y)$, we employ a circuit consisting of two stages.  First, compose a circuit of depth $\log_{3/2}(m)+O(1)$ with 3-input 2-output Full Adder gates $\mathsf{FA}(a,b,c):=(\mathsf{MAJ}(a,b,c),\mathsf{XOR}(a,b,c))$ by grouping as many triples of digits of same significance at each step as possible 
(note that $\mathsf{MAJ}$ and $\mathsf{XOR}$ are implemented in parallel).  We finish this first stage when the output contains two $\log(m)$-digit integer numbers $u$ and $v$ such that $W=u\,{+}\,v$.  To analyze this circuit, it is convenient to group all bits needing to be added into the smallest set of integer numbers, and count the reduction in the number of integers left to be added by treating layers of $\mathsf{FA}$ gates as Carry-Save Adders \cite{avizienis1961signed, wegener1987complexity}.  A Carry-Save Adder is defined as the 3-integer into 2-integer adder, which is implemented by applying the Full Adders to the individual components of the three integer numbers at the input.  
Since the number of integers left to be added changes by a factor of $\frac{2}{3}$ at each step, and every step is implemented by a depth-$1$ $\mathsf{MAJ}/\mathsf{XOR}$ circuit, the depth of the first stage is $\log_{3/2}(m)\,{+}\,O(1)$.  To find the individual components of $W(y)$, the second stage adds two $\log(m)$-digit integer numbers $u$ and $v$.  This can be accomplished by any logarithmic-depth integer addition circuit in depth $O(\log \log(m))$, such as \cite{krapchenko1970asymptotic}.  
The total depth is thus $\log_{3/2}(m)\,{+}\,O(\log \log(m))$.


Next, construct $S_5$-programs computing the $\mathsf{MAJ}$ and $\mathsf{XOR}$ functions:
\begin{eqnarray}
\label{eq:idMAJ}
\langle z_1, (1,4,3,2,5), e \rangle \;\;
\langle z_2, (1,3,5,4,2), e \rangle \;\;
\langle z_3, (1,2,5,3,4), e \rangle \;\;
\langle z_1, (1,2,3,4,5), e \rangle \nonumber \\
\langle z_2, (1,2,4,5,3), e \rangle \;\;
\langle z_3, (1,4,3,5,2), e \rangle \;\;
\langle z_1, (1,5,4,3,2), e \rangle \;\;
\langle z_1, (1,5,2,3,4), e \rangle \nonumber \\
=
\left\{ \ba{rcl}
e &\mbox{if}& \mathsf{MAJ}(z_1,z_2,z_3)=0\\
(1,2,3,4,5)  &\mbox{if}&\mathsf{MAJ}(z_1,z_2,z_3)=1, \\
\ea\right.
\end{eqnarray}


\begin{eqnarray}
\label{eq:idXOR}
\langle z_2, (1,2,3,5,4), e \rangle \;\;
\langle z_3, (1,2,4,5,3), e \rangle \;\;
\langle z_2, (1,3,5,4,2), e \rangle \;\;
\langle z_3, (1,4,5,3,2), e \rangle \nonumber \\
\langle z_1, (1,2,3,4,5), e \rangle \;\;
\langle z_2, (1,3,4,2,5), e \rangle \;\;
\langle z_2, (1,3,2,4,5), e \rangle \;\;
\langle z_3, (1,3,4,2,5), e \rangle \nonumber \\
\langle z_3, (1,3,2,4,5), e \rangle \nonumber \\
=
\left\{ \ba{rcl}
e               & \mbox{if} & \mathsf{XOR}(z_1,z_2,z_3)=0\\
(1,2,3,4,5)         & \mbox{if} & \mathsf{XOR}(z_1,z_2,z_3)=1.\\
\ea\right.
\end{eqnarray}

The branching program that $(1,2,3,4,5)$-computes $f_k(y)$ is created by recursively replacing gates $\mathsf{MAJ}$ and $\mathsf{XOR}$ in the circuit constructed above with the branching programs \eq{idMAJ} and \eq{idXOR}, where each $z_i$ is either one of the primary input variables $y_1,y_2\ldots,y_m$ or one of the intermediate variables in the circuit computing $f_k(y)$, until all instructions are controlled by constants and primary variables $y_1,y_2\ldots,y_m$.
The recoding of branches of the program $\tau$-computing a desired intermediate variable $z_*$ when $\tau{\neq}(1,2,3,4,5)$ (note how \eq{idMAJ} and \eq{idXOR} $(1,2,3,4,5)$-compute the gates, but not $\tau$-compute them for arbitrary $\tau$) is accomplished in accordance with \cite[Lemma 1]{barrington1989bounded}.
The total length of the branching program is thus upper bounded by the size of longest branching program implementation of the basic gates used ($\mathsf{MAJ}$ and $\mathsf{XOR}$) raised to the power the depth of the circuit it encodes, 
$$9^{\log_{3/2}(m)+O(\log\log(m))} = O(m^{5.4190225...}\log(m)^{O(1)}) = O(m^{5.42}).$$
\end{proof}

We conclude this section by summarizing the main result in a Theorem.

\begin{theorem}
The $n$-bit $\mathbf{hwb}$ function can be implemented by an ancilla-free reversible circuit of size $O(n^{6.42})$.
\end{theorem}
\begin{proof}
First, implement each instruction $\langle z_*, (a_1,a_2,a_3,a_4,a_5), e \rangle$ where $z_*$ is either a primary variable or a constant and the sets $\{a_1,a_2,a_3,a_4,a_5\}$ are defined per \eq{Ms}, using constantly many basic reversible gates.  This can be accomplished by employing a reversible logic synthesis algorithm, e.g., \cite{saeedi2010reversible}.  Next, use \lem{S5} with $m=n{-}5$ and $x = y\,{\sqcup}\,B$ to implement all necessary $C5|_{M_i}(f_k(y);B)$ gates, using a branching program with 
$$O\left(9^{\log_{3/2}(n-5)\,+\,O(\log\log(n-5))}\right)=O\left(9^{\log_{3/2}(n)\,+\,O(\log\log(n))}\right)$$
instructions.  Each such branching program requires $O(9^{\log_{3/2}(n)+O(\log\log(n))})$ basic reversible gates since every instruction requires constantly many basic reversible gates.  Use six \linebreak $C5|_{M_i}(f_k(y);B)$ gates to implement one $C5(f_k(x);B)$ gate, using \eq{GBk}.  Each $C5(f_k(x);B)$ thus costs
$O(9^{\log_{3/2}(n)+O(\log\log(n))})$ basic reversible gates.  Combine \lem{hwbwithc}, \lem{C0}, and \lem{lemC} to implement $\mathbf{hwb}$ using $O(n \log(n))$ $C5(f_k(x);B)$ gates, implying the total basic reversible gate count of 

$$O\left(9^{\log_{3/2}(n)\,+\,O(\log\log(n))}\cdot n \log(n)\right) = O\left(n^{6.4190225...}\log(n)^{O(1)}\right) = O(n^{6.42}).$$
\end{proof}
\section{Ancilla-free quantum circuit of size $O(n^2)$} 
\label{sec:quantumnsquared}

Consider a register of $n$ qubits and let $\mathsf{C}$ be the cyclic shift operator,
\[
\mathsf{C}|x_1,x_2,\ldots,x_{n-1},x_n\ra=|x_2, x_3, \ldots,x_n, x_1\ra.
\]
The hidden weighted bit function $U_{\mathbf{hwb}}$ may be written as
\be
\label{U}
U_{\mathbf{hwb}}|x\ra = \mathsf{C}^{x_1+x_2+\ldots+x_n}|x\ra \quad \mbox{for all $x\in \{0,1\}^n$}.
\ee
In other words, $U_{\mathbf{hwb}}$ implements the $k$-th power of $\mathsf{C}$ on the subspace with the Hamming weight $k$. Here we show that $U_{\mathbf{hwb}}$ can be implemented by an ancilla-free quantum circuit of the size $O(n^2)$.  The circuit is expressed using Clifford gates and single-qubit $Z$-rotations.

Let 
\be
\label{W}
W = \sum_{j=0}^{n-1}  |1\ra\la 1|_j
\ee
be the Hamming weight operator. Our starting point is 

\begin{lemma}
\label{lem:1}
Suppose $\mathsf{C}=e^{iH}$ for some $n$-qubit Hamiltonian $H$
that commutes with $W$. Then 
\be
\label{eq1}
U_{\mathbf{hwb}}=e^{iHW}.
\ee
\end{lemma}
\begin{proof}
Indeed, let $\calL_k$ be the subspace spanned by all basis states $|x\ra$ with the Hamming weight $k$. The full Hilbert space of $n$ qubits is the direct sum $\calL_0\oplus\calL_1\oplus \ldots\oplus\calL_n$.  Let us say that an operator $O$ is block-diagonal if $O$ maps each subspace $\calL_k$ into itself.  Since $H$ commutes with $W$, we infer that $H$ is block-diagonal.  Therefore $HW$ and $e^{iHW}$ are also block-diagonal.  Note that $HW$ and $kW$ have the same restriction onto $\calL_k$.  Thus $e^{iHW}$ and $e^{ikH}$ have the same restriction onto $\calL_k$.  By assumption, $e^{iH}=\mathsf{C}$. Thus $e^{iHW}$ and $\mathsf{C}^k$ have the same restriction onto $\calL_k$.  Likewise, $U_{\mathbf{hwb}}$ is block-diagonal and the restriction of $U_{\mathbf{hwb}}$ onto $\calL_k$  is $\mathsf{C}^k$. We conclude that $U_{\mathbf{hwb}}$ and $e^{iHW}$ have the same restriction onto $\calL_k$ for all $k$.  Since both operators are block-diagonal, one has $U_{\mathbf{hwb}}=e^{iHW}$.
\end{proof}
 
We will construct a Hamiltonian $H$ satisfying conditions of \lem{1} using the language of fermions and the fermionic Fourier transform~\cite{babbush2017low,kivlichan2018quantum}.
First, define  fermionic creation and annihilation operators $a_p^\dag$ and $a_p$ with $p\in \ZZ_n\equiv \{0,1,\ldots,n-1\}$ as
\[
a_p^\dag = \underbrace{Z\otimes Z \otimes \cdots \otimes Z}_{p} 
\otimes  |1\ra\la 0|  \otimes \underbrace{I \otimes I \otimes \cdots \otimes I}_{n-p-1}
\]
\[
a_p = \underbrace{Z\otimes Z \otimes \cdots \otimes Z}_{p} 
\otimes  |0\ra\la 1|  \otimes \underbrace{I \otimes I \otimes \cdots \otimes I}_{n-p-1}.
\]
Here $Z=|0\ra\la 0|-|1\ra\la 1|$ is the Pauli-$Z$ operator.
\begin{dfn}
A Fermionic Fourier Transform is a unitary $n$-qubit operator $\mathsf{F}$ such that $\mathsf{F}|0^n\ra =|0^n\ra$ and
\be
\label{FFT}
\mathsf{F} a_p \mathsf{F}^\dag = \frac1{\sqrt{n}} \sum_{q\in \ZZ_n} e^{2\pi i pq/n} a_q
\qquad \mbox{for all $p\in \ZZ_n$}.
\ee
\end{dfn}
Note that Eq.~(\ref{FFT}) uniquely specifies $\mathsf{F}$.  Indeed, suppose $x\in \{0,1\}^n$ is a weight-$k$ basis state with ones at qubits $p_1<p_2<\ldots<p_k$. 
Then
\be
\label{Fx}
\mathsf{F}|x\ra 
= \mathsf{F} a_{p_1}^\dag a_{p_2}^\dag \cdots a_{p_k}^\dag |0^n\ra = \mathsf{F} a_{p_1}^\dag \mathsf{F} a_{p_2}^\dag \cdots a_{p_k}^\dag  \mathsf{F} ^\dag |0^n\ra 
= \prod_{i=1}^k  \mathsf{F} a_{p_i}^\dag \mathsf{F}^\dag |0^n\ra.
\ee
Since each operator  $\mathsf{F} a_{p_i}^\dag \mathsf{F}^\dag=(\mathsf{F} a_{p_i} \mathsf{F}^\dag)^\dag$ is determined by Eq.~(\ref{FFT}), this uniquely specifies the action of $\mathsf{F}$ on the basis vectors $|x\ra$.  It will be important that $\mathsf{F}$ commutes with the Hamming weight operator $W$,
\be
\label{FW}
\mathsf{F}W=W\mathsf{F}.
\ee
Indeed, from Eqs.~(\ref{FFT},\ref{Fx}) one can see that $\mathsf{F}|x\ra$ is a linear combination of states $a_{q_1}^\dag a_{q_2}^\dag \cdots a_{q_k}^\dag |0^n\ra$. Since $(a_q^\dag)^2=0$, the state $a_{q_1}^\dag a_{q_2}^\dag \cdots a_{q_k}^\dag |0^n\ra$ is non-zero only if all indices $q_1,q_2, \ldots,q_k$ are distinct.  Such state has weight $k$.  Thus $\mathsf{F}$ maps weight-$k$ states to linear combinations of weight-$k$ states proving Eq.~(\ref{FW}).

We will use the following fact established by Kivlichan et al.~\cite{kivlichan2018quantum}.
\begin{lemma}
\label{lem:FFT}
The fermionic Fourier transform  $\mathsf{F}$ on $n$ qubits can be implemented by a quantum circuit of size $O(n^2)$. The circuit requires no ancillary qubits.
\end{lemma}
For completeness, we provide a simplified proof of \lem{FFT} and an explicit construction of the quantum circuit realizing $\mathsf{F}$ in Appendix~A. Now we are ready to define a Hamiltonian $H$ satisfying conditions of \lem{1}.
Let 
\[
E = \frac12(I + Z^{\otimes n})
\]
be the projector onto the even-weight subspace.
Define $n$-qubit Hamiltonians
\be
\label{H0}
H_0 =  \frac{2\pi}{n} \sum_{p\in \ZZ_n} p |1\ra\la 1|_p,
\qquad H' = H_0+  \frac{\pi}{n}  WE,
\ee
and
\be
\label{Hdefinition}
H = V^\dag H' V \quad \mbox{where} \quad  V =  \mathsf{F}^\dag e^{i H_0 E/2}.
\ee
\begin{lemma}
\label{lem:2}
The Hamiltonian $H$ defined in Eq.~(\ref{Hdefinition})  satisfies $C=e^{iH}$.
\end{lemma}

A proof of this lemma is given in Appendix~B. A high-level intuition behind the definition of $H$ comes from the fact that $\mathsf{F} H_0 \mathsf{F}^\dag$ is the fermionic momentum operator.  Note that $H=\mathsf{F} H_0 \mathsf{F}^\dag$ in the odd-weight subspace where $E\,{=}\,0$.  The extra terms in the definition of $H$ are needed to change integer momentums (periodic boundary conditions) in the odd-weight subspace to half-integer momentums (anti-periodic boundary conditions) in the even-weight subspace. This accounts for the difference between the qubit cyclic shift and its fermionic analogue, as detailed in Appendix~B.

From Eq.~(\ref{FW}) one can see that $HW\,{=}\,WH$. Thus $H$ satisfies conditions of \lem{1}. Combining \lem{1}, \lem{2}, and noting that $VW\,{=}\,WV$ one arrives at
\be
\label{Ucircuit}
U_{\mathbf{hwb}} = e^{i HW} 
= e^{i V^\dag H' V W} = V^\dag e^{i H'  W} V
=e^{-iH_0 E/2}  \mathsf{F} e^{i H'  W} \mathsf{F}^\dag e^{iH_0 E/2}.
\ee
Here we used the well-known fact that $e^{iV^\dag OV}=V^\dag e^{iO} V$ for any Hermitian operator $O$ and any unitary $V$ (which can be verified by expanding the exponent using the Taylor series and noting that $(V^\dag OV)^p=V^\dag O^p V$ for all $p\ge 1$). We claim that each term in Eq.~(\ref{Ucircuit}) can be implemented using $O(n^2)$ two-qubit gates without ancillary qubits.  By \lem{FFT}, the layers $\mathsf{F}$ and $\mathsf{F}^\dag$ have gate cost $O(n^2)$. 

For the term $e^{iH_0E/2}$ and its inverse, we have the following lemma.
\begin{lemma}
\label{lem:H_0E}
The operator $e^{iH_0E/2}$ can be implemented by a quantum circuit of size $O(n)$ without using ancillary qubits.
\end{lemma}
\begin{proof}
If we set $\theta_p=p\pi/n$, then
\be
\label{eq:exponentiate_H_0E}
e^{iH_0 E/2}=R_1R_2\cdots R_{n-1},\quad \mbox{where }
R_p = e^{i \theta_p |1\ra\la 1|_p E}.
\ee
The operator $|1\ra\la 1|_p E$ projects the subset of qubits $\ZZ_n{\setminus} \{p\}$ onto the odd-weight subspace. Note $p{\neq}0$ and let $C_p$ be a CNOT circuit that computes the parity of $\ZZ_n{\setminus} \{p\}$ into the qubit $0$,
\[
C_p = \prod_{j\in \ZZ_n\setminus\{0,p\}} \cnotgate_{j,0}.
\]
Then
$|1\ra\la 1|_p E = C_p^\dag |11\ra\la 11|_{0p} C_p$ and thus 
\[
R_p = C_p^\dag e^{i \theta_p |11\ra\la 11|_{0p}} C_p.
\]
Therefore, an individual $R_p$ is implemented with $O(n)$ gates, which suggests $e^{iH_0E/2}$ can be implemented with $O(n^2)$ gates. However, we can improve this count by noting that for $p{\neq} q$
\be
C_pC^\dag_q=\cnotgate_{p,0}\cnotgate_{q,0}.
\ee
Thus, in fact, the product in Eq.~\eqref{eq:exponentiate_H_0E} can be implemented with just $O(n)$ gates.
\end{proof}

We still need to implement the term $e^{iH' W} = e^{iH_0W} e^{i(\pi/n) W^2 E}$.  The operator $e^{iH_0W}$ is a product of $O(n^2)$ rotations $e^{i\theta|11\ra\la 11|}$ and $e^{i\theta |1\ra\la 1|}$.  Although a naive implementation of $e^{i(\pi/n) W^2 E}$ requires $O(n^3)$ gates, we next show that a better implementation exists.

\begin{lemma}
\label{lemma:3}
The operator $e^{i(\pi/n) W^2 E}$ can be implemented by a quantum circuit of size $O(n^2)$ without using ancillary qubits.
\end{lemma}
\begin{proof}
First, note that
\be
\label{eq:expand_W^2E}
W^2E=2\sum_{\substack{p,p'\in\ZZ_n\\ 0<p<p'}}|11\ra\la11|_{pp'}E+2\sum_{\substack{p\in\ZZ_n\\ 0<p}}|11\ra\la11|_{0p}E+\sum_{p\in\ZZ_n}|1\ra\la1|_pE.
\ee
The terms in Eq.~\eqref{eq:expand_W^2E} commute. Therefore, we have, with arbitrary order within the products,
\be
\label{eq:exponentiate_W^2E}
e^{i(\pi/n) W^2 E}=\prod_{\substack{p,p'
\in\ZZ_n \\0<p<p'}}U_{pp'}\prod_{\substack{p
\in\ZZ_n \\0<p}}U_{0p}\prod_{p\in\ZZ_n}U_p,
\ee
where, for $p\,{<}\,p'$,
\[
U_{pp'}=e^{i(2\pi/n)|11\ra\la11|_{pp'}E}\text{ \;and\; } U_{p}=e^{i(\pi/n)|1\ra\la1|_pE}.
\]

The second and third products in Eq.~\eqref{eq:exponentiate_W^2E} can be implemented with $O(n)$ gates using arguments similar to those in \lem{H_0E}. In the rest of this proof we focus on the first product and show that it can be implemented with $O(n^2)$ gates.

Notice that $|11\ra\la11|_{pp'}E$ projects the subset of qubits $\ZZ_n{\setminus}\{p,p'\}$ onto the even weight subspace while projecting qubits $p$ and $p'$ to $|11\ra_{pp'}$.  Therefore, if $0\,{<}\,p\,{<}\,p'$, we can define $S_{pp'}:=\ZZ_n{\setminus}\{0,p,p'\}$ and
\[
C_{pp'}:=\prod_{j\in S_{pp'}}\cnotgate_{j,0},
\]
such that
\[
U_{pp'}=C^\dag_{pp'}e^{i(2\pi/n)|011\ra\la011|_{0pp'}}C_{pp'}.
\]
This implementation of $U_{pp'}$ takes $O(n)$ gates, which suggests $O(n^3)$ gates might be needed to implement all $n(n{-}1)/2$ factors in the first product in Eq.~\eqref{eq:exponentiate_W^2E}.  However, we can order the factors in such a way as to allow massive cancellation between consecutive $\cnotgate$ circuits $C_{pp'}$ and implement the first product with just $O(n^2)$ total gates.

Notice that
\[
C_{pp'}C_{qq'}^\dag = \prod_{j\in S_{pp'}\Delta S_{qq'}}\cnotgate_{j,0}
\]
is a circuit of at most four $\cnotgate$ gates. In fact, it is a circuit with just two $\cnotgate$ gates when $|\{p,p'\}\cap\{q,q'\}|=1$. Thus, the following two products can be implemented with $O(n)$ gates:
\begin{align*}
U_{x\uparrow}&=U_{x,x+1}U_{x,x+2}\cdots U_{x,n-1},\\
U_{y\downarrow}&=U_{y,n-1}U_{y,n-2}\cdots U_{y,y+1},
\end{align*}
where $x,y\in\ZZ_n{\setminus}\{n{-}1\}$. Hence, the first product in Eq.~\eqref{eq:exponentiate_W^2E} can be implemented with $O(n^2)$ gates because
\[
\prod_{\substack{p,p'
\in\ZZ_n \\0<p<p'}}U_{pp'}=U_{1\uparrow}U_{2\downarrow}U_{3\uparrow}\cdots U_{n-2,n-1}.
\]
\end{proof}
The above implementation of $e^{i(\pi/n) W^2 E}$ requires three-qubit gates of the form $e^{i\theta |011\ra\la 011|}$. The latter can be decomposed into a sequence of $O(1)$ two-qubit Clifford gates and single-qubit $Z$-rotations using the standard methods~\cite{nielsen2002quantum}. We summarize main result of this section in the following Theorem. 

\begin{theorem}
Eq.~(\ref{Ucircuit}) reports an ancilla-free quantum circuit of size $O(n^2)$ implementing $U_{\mathbf{hwb}}$.
\end{theorem}

\section{Conclusion}
In this paper, we introduced two ancilla-free circuits implementing the Hidden Weighted Bit function, $O(n^{6.42})$-gate reversible circuit and $O(n^2)$-gate quantum circuit.  Our circuits improve best previously known exponential size reversible and quantum ancilla-free circuits into polynomial-size ones.  Our results demote $\mathbf{hwb}$ by removing it from the class of ``hard'' benchmarks \cite{saeedi2013synthesis}.  Our ancilla-free reversible implementation marks a new point in the study of ancilla vs gate count (space-time) tradeoff.  Noting a high exponent in the reversible circuit complexity and a more-than-qubic difference between complexities of our best quantum and reversible circuit implementations, we suggest that a further line of inquiry may target improving the reversible implementation.

\section*{Acknowledgements}
SB and TY are partially supported by the IBM Research Frontiers Institute.

\section*{Appendix A}
\label{app:A}

In this Appendix we construct a quantum circuit  implementing the fermionic Fourier transform $\mathsf{F}$ on $n$ qubits and illustrate it for $n{=}3$. The circuit is expressed using $O(n^2)$ single-qubit and two-qubit gates
\[
\mathsf{S}(\gamma) = e^{i \gamma |1\ra\la 1|} 
=\left[ \ba{cc} 1 & 0 \\ 0 & e^{i\gamma} \\ \ea \right]
\]
and
\[
\mathsf{R}(\alpha,\beta) = e^{\alpha e^{i\beta} |10\ra\la 01| - \alpha e^{-i\beta} |01\ra\la 10|}
=\left[
\ba{cccc}
1 & 0 & 0 & 0 \\
0 & \cos{(\alpha)} & -e^{-i\beta} \sin{(\alpha)} & 0 \\
0 & e^{i\beta} \sin{(\alpha)} & \cos{(\alpha)} & 0 \\
0 & 0 & 0 & 1 \\ \ea \right].
\] 
Here $\alpha,\beta,\gamma$ are real parameters. We use subscripts $p,q\in\ZZ_n$ to indicate qubits acted upon by each gate. In the fermionic language, $\mathsf{R}_{p,p+1}(\alpha,\beta)$ implements a Givens rotation 
in the two-dimensional subspace spanned by operators $a_p$ and $a_{p+1}$.
Namely, let $\mathsf{R}_{p,p+1}= \mathsf{R}_{p,p+1}(\alpha,\beta)$. Then
\begin{align}
\mathsf{R}_{p,p+1} a_p \mathsf{R}_{p,p+1}^\dag =&
\cos{(\alpha)} a_p - \sin{(\alpha)} e^{i\beta} a_{p+1}, \label{Givens1} \\
\mathsf{R}_{p,p+1} a_{p+1} \mathsf{R}_{p,p+1}^\dag =&
\sin{(\alpha)} e^{-i\beta} a_p + \cos{(\alpha)}  a_{p+1}, \label{Givens2}
\end{align}
We also need a fermionic SWAP gate~\cite{kivlichan2018quantum, verstraete2009quantum} defined as
\[
\mathsf{fSWAP} = \mathsf{CZ} \cdot \mathsf{SWAP}= \mathsf{R}(\pi/2,\pi/2) \mathsf{S}(-\pi/2)^{\otimes 2}.
\]
One can easily check that
\[
(\mathsf{fSWAP}_{p,p+1}) a_p (\mathsf{fSWAP}_{p,p+1})^\dag = a_{p+1}
\quad \mbox{and} \quad  (\mathsf{fSWAP}_{p,p+1}) a_{p+1} (\mathsf{fSWAP}_{p,p+1})^\dag = a_p.
\]
Define a unitary $n\,{\times}\, n$ matrix $f$ with matrix elements
\be
\label{Fnxn}
f_{p,q} = n^{-1/2} e^{ 2\pi i p q/n}, \text{ \;where }p,q\in \ZZ_n.
\ee
We will write $\mathsf{row}(f,p)$ for the $p$-th row of $f$. Below we define a function $\mathsf{ColumnReduce}(f,m,U)$ that takes as input a unitary $n\times n$ matrix $f$, an integer $m\in \ZZ_n$, and a quantum circuit $U$ acting on $n$ qubits.  The function returns a modified unitary matrix $f'$ and a modified quantum circuit $U'$.  A quantum circuit realizing the fermionic Fourier transform $\mathsf{F}$ on $n$ qubits
is generated by the following algorithm.

{\centering
\begin{minipage}{1.0\linewidth}
\begin{algorithm}[H]
\caption{$\quad \mathsf{Fermionic Fourier Transform}$}
	\begin{algorithmic}[1]
		\State{Let $f$ be the $n\times n$ unitary matrix defined in Eq.~(\ref{Fnxn})}
		\State{$U\gets I$\Comment{Empty quantum circuit}}
		\For{$m=n-1$ to $0$}
		\State{$(f,U)=\mathsf{ColumnReduce}(f,m,U)$}
		\EndFor
		\State{\Return $\mathsf{F}=U^{-1}$}		
		\end{algorithmic}
\end{algorithm}
\end{minipage}
}

{\centering
\begin{minipage}{1.0\linewidth}
\begin{algorithm}[H]
\caption{$\quad \mathsf{ColumnReduce}(f,m,U)$}
	\begin{algorithmic}[1]
		\For{$p=0$ to $m-1$}
		\If{$f_{p,m}\ne 0$ $\mathbf{or}$ $f_{p+1,m}\ne 0$}
		\If{$f_{p+1,m}=0$}
		\State{Swap $\mathsf{row}(f,p)$ and $\mathsf{row}(f,p+1)$}
		\State{$U\gets \mathsf{fSWAP}_{p,p+1} \cdot U$\Comment{Add $\mathsf{fSWAP}$ gate}}		
		\EndIf
		\Comment{Now $f_{p+1,m}\ne 0$}
		\State{Choose angles $\alpha,\beta$ such that $\tan{(\alpha)}e^{-i\beta} = - f_{p,m}/f_{p+1,m}$}
				\State{$v\gets \mathsf{row}(f,p)$}

		\State{$\mathsf{row}(f,p) \gets   \cos{(\alpha)}\mathsf{row}(f,p)  + \sin{(\alpha)} e^{-i\beta} \mathsf{row}(f,p+1)$\Comment{Now $f_{p,m}=0$}}
		\State{$\mathsf{row}(f,p+1) \gets \cos{(\alpha)}
		\mathsf{row}(f,p+1) - \sin{(\alpha)} e^{i\beta} v$}
		\State{$U\gets \mathsf{R}_{p,p+1}(\alpha, \beta)\cdot U$\Comment{Add $\mathsf{R}$ gate}}
		\EndIf
		\EndFor
		\Comment{Now $f_{m,m}$ is the only nonzero in the $m$-th column of $f$}
		\State{$\gamma \gets \mathrm{phase}(f_{m,m})$ \Comment{Now $f_{m,m}=e^{i\gamma}$}}
		\State{$f_{m,m}=1$}
		\State{$U\gets \mathsf{S}_m(\gamma) \cdot U$\Comment{Add $\mathsf{S}$ gate}}	
		\State{\Return{$(f,U)$}}
		\end{algorithmic}
\end{algorithm}
\end{minipage}
}

We claim that the quantum circuit $U$ and the unitary matrix $f$ obtained after each call to the function $\mathsf{ColumnReduce}$ have the property
\be
\label{algo1}
(U\mathsf{F}) a_p (U\mathsf{F})^\dag = \sum_{q\in \ZZ_n}^n f_{p,q} a_q
\mbox{ \;for all $p\in \ZZ_n$}.
\ee
Indeed, Eq.~(\ref{algo1}) is trivially true initially when $U{=}I$ and $f$ is defined by Eq.~(\ref{Fnxn}). The lines 4 and 7-10 of Algorithm~2 apply a sequence of Givens rotations to the matrix $f$ setting to zero all matrix elements $f_{p,m}$ with $0\,{\le}\,p\,{<}\,m$ and setting $f_{m,m}{=}1$.  The order in which matrix elements of $f$ are set to $0$ or $1$ is illustrated for $n{=}3$ below (asterisks indicate matrix elements of $f$).
\vspace{3mm}

\resizebox{15cm}{0.6cm}{
$\displaystyle{
\left[ \ba{ccc}
* & * & * \\
* & * & * \\
* & * & * \\
\ea\right]
\to
\left[ \ba{ccc}
* & * & 0 \\
* & * & * \\
* & * & * \\
\ea\right]
\to
\left[ \ba{ccc}
* & * & 0\\
* & * & 0 \\
* & * & * \\
\ea\right]
\to
\left[ \ba{ccc}
* & * & 0 \\
* & * & 0\\
* & * & 1 \\
\ea\right]
\to
\left[ \ba{ccc}
* & 0 & 0\\
* & * & 0\\
* & * & 1 \\
\ea\right]
\to
\left[ \ba{ccc}
* & 0 & 0\\
* & 1 & 0\\
* & * & 1 \\
\ea\right]
\to
\left[ \ba{ccc}
1 & 0 & 0\\
* & 1 & 0\\
* & * & 1 \\
\ea\right]
}$}

\vspace{3mm}
\noindent
Since $f$ remains unitary at each step, the final unit-diagonal low-triangular matrix  is the identity, i.e. $f{=}I$ after the last iteration of Algorithm~1.  Each time a Givens rotation is applied to some rows $p,p+1$ of the matrix $f$, the corresponding Givens rotations of fermionic operators $a_p$, $a_{p+1}$  are added to the quantum circuit $U$, see  Eqs.~(\ref{Givens1},\ref{Givens2}).  More precisely, the angles $\alpha,\beta$ at Line~7 are chosen such that the operator 
\[
\mathsf{R}_{p,p+1}(\alpha,\beta) (f_{p,m} a_p + f_{p+1,m} a_{p+1} ) \mathsf{R}_{p,p+1}(\alpha,\beta)^\dag
\]
is proportional to $a_{p+1}$, see Eqs.~(\ref{Givens1},\ref{Givens2}). Thus the property Eq.~(\ref{algo1}) is maintained at each step. After the last iteration of Algorithm~1 one has $f=I$ and Eq.~(\ref{algo1}) gives $(U\mathsf{F}) a_p (U\mathsf{F})^\dag =a_p$ for all $p$. Furthermore $U|0^n\ra=|0^n\ra$ since all gates added to $U$ map $|0^n\ra$ to itself. We conclude that $U=\mathsf{F}^{-1}$ after the last iteration of Algorithm~1. Thus the algorithm returns a quantum circuit realizing $\mathsf{F}$.  The inverse circuit $U^{-1}$ can be obtained from $U$ using the identities $\mathsf{R}(\alpha,\beta)^{-1}=\mathsf{R}(-\alpha,\beta)$ and $\mathsf{S}(\gamma)^{-1}=\mathsf{S}(-\gamma)$.  The direct inspection shows that the total numberof gates $\mathsf{fSWAP}$, $\mathsf{R}$ and $\mathsf{S}$ added to $U$ is $O(n^2)$.  We implemented Algorithms~1,2 in Matlab obtaining the following circuit in the case $n{=}3$.

\begin{figure}[h]
\centerline{\includegraphics[height=4.5cm]{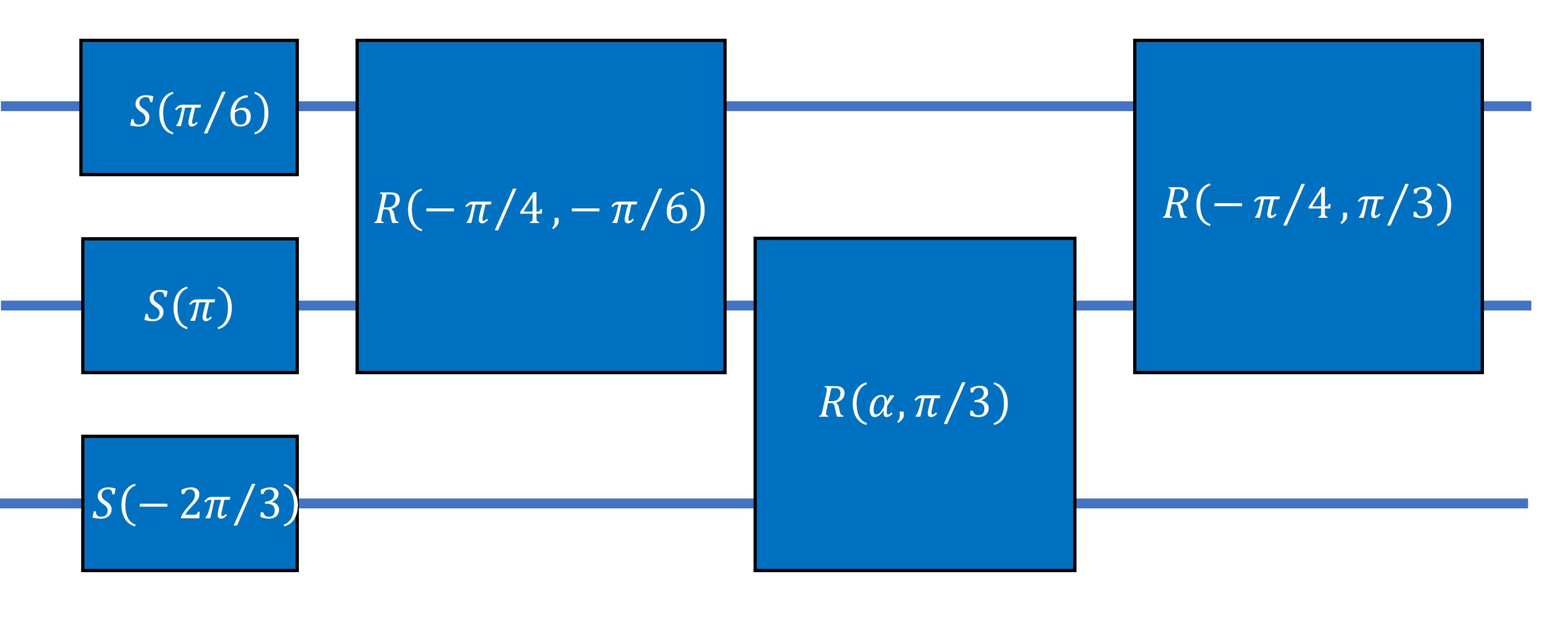}}
\caption{Quantum circuit realizing the $3$-qubit fermionic Fourier transform $\mathsf{F}$.
The circuit was generated using Algorithm~1.
Here $\alpha = - (1/2) \arccos{(1/3)} \approx -0.9553$.}
\end{figure}

\section*{Appendix B}
\label{app:B}

Here we prove \lem{2}. First note that
\[
\mathsf{C} =\mathsf{SWAP}_{0,1} \mathsf{SWAP}_{1,2}\cdots \mathsf{SWAP}_{n-2,n-1}.
\]
Define a fermionic SWAP operator~\cite{kivlichan2018quantum,verstraete2009quantum}
\be
\label{fSWAP1}
\mathsf{fSWAP}= \mathsf{CZ}\cdot \mathsf{SWAP}
=\left[\ba{cccc}
1 & 0 & 0 & 0 \\
0 & 0 & 1 & 0 \\
0 & 1 & 0 & 0 \\
0 & 0 & 0 & -1 \\
\ea\right]
\ee
and a fermionic cyclic shift
\be
\label{fC1}
\mathsf{fC} =  \mathsf{fSWAP}_{0,1} \mathsf{fSWAP}_{1,2}\cdots \mathsf{fSWAP}_{n-2,n-1}.
\ee

A simple algebra shows that
\be
\label{fC2}
\mathsf{C}|x\ra = (-1)^{x_{n-1}(x_0+x_1+\ldots+x_{n-2})} \mathsf{fC}|x\ra.
\ee
Let $k=x_0+x_1+\ldots+x_{n-1}$ be the  Hamming weight of $x$.
Then 
\be
\label{phase1}
(-1)^{x_{n-1}(x_0+x_1+\ldots+x_{n-2})} = (-1)^{x_{n-1}(k-x_{n-1})}=(-1)^{x_{n-1}k +x_{n-1}} = (-1)^{x_{n-1}(k+1)}.
\ee
Thus $\mathsf{C}\,{=}\,\mathsf{fC}$ on the odd-weight subspace
and $\mathsf{C}\,{=}\,\mathsf{fC}Z_{n-1}$ on the even-weight subspace, i.e.
\be
\label{Cparts}
\mathsf{C}= E \mathsf{fC}Z_{n-1} + (I-E) \mathsf{fC}.
\ee
We claim that 
\be
\label{identity1}
\mathsf{fC} = \mathsf{F} e^{iH_0}  \mathsf{F}^\dag.
\ee
Indeed, let 
\[
G\equiv  \mathsf{F} e^{iH_0}  \mathsf{F}^\dag.
\]
First note that $\mathsf{fC}|0^n\ra = G|0^n\ra = |0^n\ra$.  Since any state can be obtained from $|0^n\ra$ by applying the creation operators $a_p^\dag$, it suffices to check that 
\[
\mathsf{fC}^\dag a_p \mathsf{fC} =   G^\dag a_p G
\]
for all $p$. A simple algebra shows  that 
\[
(\mathsf{fSWAP}_{p,p+1}) a_p (\mathsf{fSWAP}_{p,p+1})^\dag = a_{p+1}
\quad \mbox{and} \quad  (\mathsf{fSWAP}_{p,p+1}) a_{p+1} (\mathsf{fSWAP}_{p,p+1})^\dag = a_p.
\]
Combining this and Eq.~(\ref{fC1}) one gets
\[
\mathsf{fC}^\dag a_p \mathsf{fC} =a_{p-1},
\]
where the indices of fermionic operators are evaluated modulo $n$.
Using the identities
\[
e^{-iH_0} a_q e^{iH_0} = e^{2\pi i q/n} a_q,
\]
\[
\mathsf{F} a_p \mathsf{F}^\dag = \frac1{\sqrt{n}} \sum_{q\in \ZZ_n}^n e^{2\pi i pq/n} a_q,
\quad \mbox{and} \quad
\mathsf{F}^\dag a_q \mathsf{F} = \frac1{\sqrt{n}} \sum_{r\in \ZZ_n}^n e^{-2\pi i qr/n} a_r,
\]
 one gets
\[
G^\dag a_p G =  n^{-1/2} \sum_{q\in \ZZ_n} e^{2\pi i (1-p)q/n} \mathsf{F} a_q\mathsf{F}^\dag
= n^{-1} \sum_{q,r\in \ZZ_n} e^{2\pi i (1-p+r)q/n} a_r = a_{p-1}.
\]
Thus $G^\dag a_p G=\mathsf{fC}^\dag a_p \mathsf{fC} =a_{p-1}$, proving Eq.~(\ref{identity1}).

Next we claim that 
\be
\label{identity2}
\mathsf{fC}Z_{n-1} = e^{-iH_0/2} \mathsf{fC}e^{iH_0/2} e^{i(\pi/n)W}.
\ee
Indeed, let 
\[
L:= \mathsf{fC}Z_{n-1}  \quad \mbox{and} \quad R:= e^{-iH_0/2} \mathsf{fC}e^{iH_0/2} e^{i(\pi/n)W}.
\]
Since $L|0^n\ra=R|0^n\ra=|0^n\ra$,
it suffices to check that $L^\dag a_p L = R^\dag a_p R$ for all
$p\in \ZZ_n$.
A simple algebra gives
\[
e^{-i(\pi/n)W}a_p e^{i(\pi/n)W} = e^{i(\pi/n)} a_p,
\]
\[
Z_{n-1} a_p Z_{n-1} = \left\{ \ba{rcl}
a_p &\mbox{if} & 0\le p\le n-2\\
-a_p & \mbox{if} & p = n-1\\
\ea
\right.,
\]
\[
e^{iH_0/2}a_p e^{-iH_0/2} = e^{-i\pi p/n} a_p,
\quad \mbox{and} \quad
e^{-iH_0/2}a_p e^{iH_0/2} = e^{i\pi p/n} a_p
\]
for all $p\in \ZZ_n$. Recall that
$\mathsf{fC}^\dag a_p \mathsf{fC} =a_{p-1}$.
Using the above identities one gets
\[
L^\dag a_p L=
\left\{ \ba{rcl}
a_p &\mbox{if} & 1\le p\le n-1\\
-a_p & \mbox{if} & p = 0\\
\ea
\right.
\]
and
\[
R^\dag  a_p R
=e^{-i\pi p/n}e^{i\pi p'/n} e^{i(\pi/n)}  a_{p-1},
\]
where $p'\equiv p-1 {\pmod n}$. Note that
\[
e^{-i\pi p/n}e^{i\pi p'/n} = \left\{ \ba{rcl}
e^{-i\pi/n} &\mbox{if} & 1\le p\le n-1\\
- e^{-i\pi/n}& \mbox{if} & p = 0\\
\ea
\right..
\]
Thus $L^\dag a_p L = R^\dag a_p R$, that is, $L{=}R$, proving Eq.~(\ref{identity2}).

Combining Eqs.~(\ref{Cparts},\ref{identity1},\ref{identity2}) one infers that 
the restrictions of $\mathsf{C}$ onto the odd-weight and even-weight subspaces coincide with the
operators
\[
\mathsf{C}_{odd}=\mathsf{F} e^{iH_0}  \mathsf{F}^\dag
\]
and
\[
\mathsf{C}_{even}= e^{-iH_0/2}(\mathsf{F} e^{iH_0}  \mathsf{F}^\dag) e^{iH_0/2} e^{i(\pi/n)W}
\]
respectively. Thus
\[
\mathsf{C}=e^{-iH_0E/2}(\mathsf{F} e^{iH_0}  \mathsf{F}^\dag) e^{iH_0E/2} e^{i(\pi/n)WE}
\]
on the full Hilbert space. Recall that the fermionic Fourier transform  $\mathsf{F}$ preserves the Hamming weight. Thus $\mathsf{F}^\dag$ commutes with $e^{i(\pi/n)WE}$. Commuting the term $e^{i(\pi/n)WE}$ to the left gives
\[
\mathsf{C}=e^{-iH_0E/2}\mathsf{F} \left(e^{iH_0}e^{i(\pi/n)WE} \right) \mathsf{F}^\dag e^{iH_0E/2}
=V^\dag e^{iH'} V,
\]
where $V=\mathsf{F}^\dag e^{iH_0E/2}$ and $H'=H_0 + (\pi/n)WE$. Thus $\mathsf{C}=e^{iV^\dag H'V}$, proving \lem{2}.

\end{document}